\pdfoutput=1
\documentclass[sigconf]{acmart}
\usepackage{amsthm}
\usepackage{algorithm}
\usepackage{algorithmic}
\usepackage{enumitem}
\newtheorem{theorem}{Theorem}
\usepackage{graphicx}
\usepackage{multirow}
\usepackage{balance}

\copyrightyear{2023}
\acmYear{2023}
\setcopyright{rightsretained}
\acmConference[KDD '23]{Proceedings of the 29th ACM SIGKDD Conference on Knowledge Discovery and Data Mining}{August 6--10, 2023}{Long Beach, CA, USA} 
\acmBooktitle{Proceedings of the 29th ACM SIGKDD Conference on Knowledge Discovery and Data Mining (KDD '23), August 6--10, 2023, Long Beach, CA, USA} 
\acmDOI{10.1145/3580305.3599536} 
\acmISBN{979-8-4007-0103-0/23/08}

\copyrightyear{2023}
\acmYear{2023}
\setcopyright{rightsretained}
\acmConference[KDD '23]{Proceedings of the 29th ACM SIGKDD Conference on Knowledge Discovery and Data Mining}{August 6--10, 2023}{Long Beach, CA, USA}
\acmBooktitle{Proceedings of the 29th ACM SIGKDD Conference on Knowledge Discovery and Data Mining (KDD '23), August 6--10, 2023, Long Beach, CA, USA}
\acmPrice{}
\acmDOI{10.1145/3580305.3599536}
\acmISBN{979-8-4007-0103-0/23/08}

\begin{document}
\title{Unbiased Delayed Feedback Label Correction \\ for Conversion Rate Prediction}

\author{Yifan Wang}
\orcid{0000-0002-2933-6363}
\affiliation{%
  \institution{DCST, BNRist, Tsinghua University}
  \city{Beijing}
  \postcode{10084}
  \country{China}
}
\email{yf-wang21@mails.tsinghua.edu.cn}

\author{Peijie Sun}
\orcid{0000-0001-9733-0521}
\affiliation{%
  \institution{DCST, BNRist, Tsinghua University}
  \city{Beijing}
  \postcode{10084}
  \country{China}
}
\email{sun.hfut@gmail.com}

\author{Min Zhang}
\authornote{Corresponding author}
\orcid{0000-0003-3158-1920}
\affiliation{%
  \institution{DCST, BNRist, Tsinghua University}
  \city{Beijing}
  \postcode{10084}
  \country{China}
}
\email{z-m@tsinghua.edu.cn}

\author{Qinglin Jia}
\orcid{0000-0002-3583-6719}
\affiliation{%
  \institution{Noah's Ark Lab, Huawei}
  \city{Beijing}
  \country{China}
}
\email{jiaqinglin2@huawei.com}

\author{Jingjie Li}
\orcid{0000-0001-5253-1899}
\affiliation{%
  \institution{Noah's Ark Lab, Huawei}
  \city{Beijing}
  \country{China}
}
\email{lijingjie1@huawei.com}

\author{Shaoping Ma}
\orcid{0000-0002-8762-8268}
\affiliation{%
  \institution{DCST, BNRist, Tsinghua University}
  \city{Beijing}
  \postcode{10084}
  \country{China}
}
\email{msp@tsinghua.edu.cn}

\renewcommand{\shortauthors}{Yifan Wang et al.}

\begin{abstract}
Conversion rate prediction is critical to many online applications such as digital display advertising. To capture dynamic data distribution, industrial systems often require retraining models on recent data daily or weekly. However, the delay of conversion behavior usually leads to incorrect labeling, which is called delayed feedback problem. Existing work may fail to introduce the correct information about false negative samples due to data sparsity and dynamic data distribution. To directly introduce the correct feedback label information, we propose an Unbiased delayed feedback Label Correction framework (ULC), which uses an auxiliary model to correct labels for observed negative feedback samples. Firstly, we theoretically prove that the label-corrected loss is an unbiased estimate of the oracle loss using true labels. Then, as there are no ready training data for label correction, counterfactual labeling is used to construct artificial training data. Furthermore, since counterfactual labeling utilizes only partial training data, we design an embedding-based alternative training method to enhance performance. Comparative experiments on both public and private datasets and detailed analyses show that our proposed approach effectively alleviates the delayed feedback problem and consistently outperforms the previous state-of-the-art methods.
\end{abstract}




\maketitle

\section{Introduction}
Predicting the probability of users clicking or converting on ads or items is critical to many online applications, such as digital display advertising and recommender systems. Take online advertising as an example. Generally, ad delivery platforms provide advertisers with several optional billing models, such as Cost Per thousand iMpressions (CPM), Cost Per Click (CPC) and Cost Per Acquisition (CPA), in which CPA is preferred as the conversion is closer to advertisers' profits. For the CPA model, predicting the click rate and conversion rate of users to the placed advertisements is the key to achieving more revenue, which are also known as the Click-Through Rate (CTR) and Conversion Rate (CVR) prediction tasks. These two tasks have received increasing attention from industry and academia in recent years \cite{din, deepfm,ctr21}.

Model freshness is important for CTR and CVR prediction models as user interests change dynamically. A common strategy to keep fresh in the industry is to retrain the model daily or weekly on all collected data. This simple strategy can be effective for CTR prediction. However, the delay of conversion behavior makes it challenging to ensure the freshness of CVR models, which is called \textit{Delayed Feedback Problem}. Unlike click behavior happening quickly within minutes of impression, conversions occur much more slowly after days, sometimes taking up to weeks \cite{201401}.
This leads that the ground truth of recently clicked but unconverted samples is unknown as they may convert in the future. 

A vanilla solution is to treat all these unconverted samples as negative feedback, which will cause some positive samples (i.e., real conversions) to be mislabeled, leading to the false negative problem. These mislabeled samples can significantly damage the performance of the CVR prediction model as they are important to the model freshness. Another obvious solution is to wait for a long time until the labels are accurate enough. However, this means that the data is old, which conflicts with the purpose of keeping models fresh. Thus, the delayed feedback problem reflects a trade-off between model freshness and label correctness.
Therefore, handling fresh unconverted data with unknown labels is an important challenge for CVR prediction.

Existing methods for the delayed feedback problem can be classified into two types based on the problem setting: online training \cite{201904,202102,202103,202106,202110,202202, nips22, calibration22} and offline training \cite{201401,201801,202006,202007}. 
In the online setting, when new user behaviors are logged, the model is immediately updated on the new data to keep it fresh.
In the offline setting, the model is retrained daily or weekly on all collected data to ensure freshness.
Industrial systems often choose the appropriate training setting based on their business requirements \cite{202106}. The methods under different settings are quite different, and here we only focus on the offline training.

To our knowledge, DFM \cite{201401} first studied the delayed feedback problem. They propose to explicitly model delay time with delay distribution assumptions. However, the actual delay may not obey the assumptions, which leads to its suboptimal performance. Recent work \cite{202006, 202204} in offline learning attempts to construct unbiased estimates of the oracle loss that uses true labels to alleviate the delayed feedback problem. However, although these methods are theoretically unbiased, we argue that they may fail to introduce the correct information about false negative samples, as they do not construct the correct samples corresponding to these mislabeled samples.

In this paper, we aim to address the delayed feedback problem in CVR prediction through label correction. We theoretically prove that if the label of the observed unconverted samples can be corrected to its probability of being a false negative sample, the label-corrected loss will be an unbiased estimate of the oracle loss that uses true labels. Compared to existing unbiased methods for delayed feedback, the advantage of label correction is that it directly complements the information of the correct samples corresponding to the false negative samples. Further, we attempt to train a label correction (LC) model to correct the labels for the observed unconverted samples. If the LC model is accurate enough, the delayed feedback problem can be well addressed.

However, how to train an accurate LC model is non-trivial. Above all, there exist no ready training data for the LC model. We use a counterfactual labeling method to construct artificial training data by setting a counterfactual deadline. Nevertheless, counterfactual labeling utilizes only partial training data. To enhance the performance of the LC model, we further designed an alternative learning method based on embedding transfer. To demonstrate the effectiveness of our method, we conduct extensive experiments on the public and private datasets. Experimental results show that our method can effectively alleviate the delayed feedback problem.
Our main contributions can be summarized as follows.
\begin{itemize}[leftmargin=*]
    \item To the best of our knowledge, it is the first work in the offline training setting to use an unbiased label correction approach to solve the delayed feedback problem of CVR prediction.
    \item We give a theoretical analysis of unbiased label learning and propose an alternative learning framework for CVR prediction to meet the delayed feedback challenge.
    \item Comparative experimental results on both public and private datasets demonstrate the effectiveness of our proposed framework.
\end{itemize}
\section{Related Work}
\subsection{CVR Prediction}
The CVR prediction task shares many similarities with the widely studied CTR prediction task. They both predict the probability of a user performing a certain behavior on an ad or an item. Besides, their inputs are generally the same. Generally, the model structure designed for the CTR task can also be applied to CVR prediction. Thus, existing research on CVR prediction focuses more on the differences between CVR and CTR.

There are three main challenges for CVR prediction. First, the data for the CVR task are often more sparse than the CTR task. Existing research mitigates this problem through multi-task learning \cite{201902,201903} and pre-training \cite{202001}. Second, CVR prediction suffers from selection bias. The CVR prediction model is trained on click samples but infers for all exposure samples during inference. Differences in exposure distribution and click distribution lead to selection bias, which existing work addresses through entire sample space modeling \cite{202003,201803,202101,ESCM2}, inverse propensity score \cite{202004,202005}, and doubly robust methods \cite{202008,202105}. Third, conversions do not happen as immediately as clicks, with some conversions taking days or even a week. This could result in some positive samples that have not yet converted being incorrectly treated as negative samples. Existing studies address it by delay time modeling \cite{201401, 201801, 202002} or importance sampling \cite{202006,202106}, which we will detail in Session 2.2.
In this work, we focus on the third challenge, the delayed feedback problem, and leave the extension of our method to other problems for future work.

\subsection{Delayed Feedback}
Here we only focus on the delayed feedback problem in the offline setting.

To our knowledge, the delayed feedback problem was first studied by DFM \cite{201401}. DFM models the delay time explicitly. It assumes that the delay time obeys an exponential distribution and then optimizes the maximum likelihood of the currently observed data labels. \cite{201801} extends this approach further by using a non-parametric approach to modeling delay time. A drawback of the above methods is that they try to optimize the observed conversion information instead of directly optimizing the true conversion information.

In contrast to explicitly modeling delay time, recent work \cite{202006, 202204} attempts to address the delay feedback problem by constructing unbiased estimates of the oracle loss that uses true labels. FSIW \cite{202006} leverages importance sampling to construct an unbiased loss.
Intuitively, it increases the weight of observed positive samples and decreases the weight of potentially negative samples as these samples may be mislabeled.
Besides, nnDF \cite{202204} assumes that the labels of samples before a time window are accurate and then uses these samples to correct for the biased loss of the whole training data.
A drawback of the above methods is that, despite their theoretical guarantee of unbias, they might fail to introduce information about the correct positive sample for each specific false negative sample. For FSIW, it only reduces the weight of the mislabeled samples but cannot introduce the information of the corrected sample, i.e., the weight of the corresponding correct sample is still zero. For nnDF, it does not process recent samples, and therefore cannot introduce information about the correct samples among them. This problem is worse when the data distribution has changed recently. As the information about the false negative samples may differ from the past observed positive samples, only using the observed positive samples cannot complement the correct information about the fresh false negative samples.

Unlike the above state-of-art approaches, we propose to correct the label for each observed negative sample so that the information of the correct sample can be introduced directly. Besides, it can also be proved theoretically that the label-corrected loss is an unbiased estimate of the oracle loss.
\section{Preliminaries}
\subsection{Notations}
In online advertising platforms, the user behaviors for the display ads are logged to train the CVR prediction model. Suppose we collect training data $\mathcal{D}$ at timestamp $T$, i.e., we can obtain all the user behaviors and corresponding features before $T$. Let $\mathcal{D} = \{\left(x_i, v_i, e_i, cts_i, cvt_i \right), i=1, 2, ...\}$. The notation $i$ denotes the $i$-th sample. Each sample represents a click record of users. For the $i$-th sample $\left(x_i, v_i, e_i, cts_i, cvt_i\right)$, $x_i$ denotes the feature information of this sample. $cts_i$ denotes the click timestamp. $v_i$ is a binary value that denotes whether the clicked ad has a further conversion before the observed timestamp $T$. If $v_i = 1$, $cvt_i$ will record the corresponding conversion timestamp. Otherwise, $cvt_i$ is empty. $e_i$ denotes the time elapsed from $cts_i$ to $T$, i.e., $T - cts_i$. 

Let $c_i$ denote whether the $i$-th sample will finally lead to a conversion. Note that we cannot wait forever for the possible conversion to happen. In practice, a long time window $w_a$ is applied depending on the specific scenario, e.g., one month for Criteo \cite{201401}. Only conversions within the time window after clicks are considered valid. In other words, if $e_i \ge w_a$, then $v_i = c_i$. If $e_i < w_a$, $c_i$ is unknown. Thus, $c_i$ is not included in the training data $\mathcal{D}$. For test data, we can wait enough time to obtain $c_i$ for evaluation.

For easy reading, the notations are summarized in Table \ref{notations}.

\begin{table}[h]
\caption{Notations and Explanations of Variables}
\label{notations}
\begin{tabular}{l|l}
\hline
\textbf{Notation} & \textbf{Explanation}         \\
\hline
$T$     & the training data collection timestamp          \\
$x_i$     & the feature vector of $i$-th sample                        \\
$v_i$      & the observed conversion label of $i$-th sample                  \\
$e_i$     & the elapsed time of $i$-th sample          \\
$c_i$     & the true conversion label of $i$-th sample          \\
$cts_i$     & the click timestamp of $i$-th sample          \\
$cvt_i$     & the conversion timestamp of $i$-th sample          \\
\hline
\end{tabular}
\end{table}

\subsection{Task Formulation}
The conversion rate is defined as the probability of the final conversion for a clicked ad, i.e., $p_{CVR} = p(c_i=1|x_i)$. The CVR prediction task under delayed feedback is aimed to use the training data $\mathcal{D}$ collected at $T$ to predict $p_{CVR}$ for the clicked ads after $T$.

Note that training samples clicked before $T - w_a$ (i.e., $e_i \ge w_a$) can be fed directly into the model without any processing as their labels are correct. Since the core issue for delayed feedback is how to handle the fresh data with unknown labels, we omit these data in the rest of this paper for simplicity, which does not influence the correctness of our proof and method.

\subsection{Vanilla and Oracle Loss}
Next, we introduce the two basic loss functions in the delayed feedback problem. Note that CVR prediction is essentially a binary classification problem. Generally, the cross-entropy loss is adopted for training the CVR model. Let $f(\cdot;\theta)$ denote the CVR model with trainable parameters $\theta$. Suppose we can now foresee the future and obtain an ideal dataset $\mathcal{D}^*$, which contains $c_i$ for each sample. Then the cross-entropy loss can be written as:
\begin{equation}
    \mathcal{L}_{oracle}(\mathcal{D}^*) = \frac{1}{|\mathcal{D}^*|} \sum_{i = 1}^{|\mathcal{D}^*|} \left(c_i\log f(x_i;\theta) + (1 - c_i) \log (1 - f(x_i;\theta))\right) ,
\end{equation}

Equation (1) is called the oracle loss $\mathcal{L}_{oracle}$ as we suppose the final conversion label $c_i$ for each click record is available. 

However, in practice, we cannot obtain the oracle label $c_i$ for each sample at the data collection timestamp $T$. If we ignore the delayed feedback and replace the oracle label $c_i$ with the observed label $v_i$, we can get the vanilla loss $\mathcal{L}_{vanilla}$ for CVR model training:
\begin{equation}
    \mathcal{L}_{vanilla}(\mathcal{D}) = \frac{1}{|\mathcal{D}|} \sum_{i = 1}^{|\mathcal{D}|} \left(v_i\log f(x_i;\theta) + (1 - v_i) \log (1 - f(x_i;\theta))\right) .
\end{equation}

Note that some samples may convert after the data collection timestamp $T$. Obviously, the vanilla loss will incorrectly treat some positive samples as negative samples, which will damage the performance of CVR prediction model.

\section{Unbiased Label Correction For Delayed Feedback Problem}
\subsection{Overall Framework}
\begin{figure*}[h]
  \centering
  \includegraphics[width=0.9\linewidth]{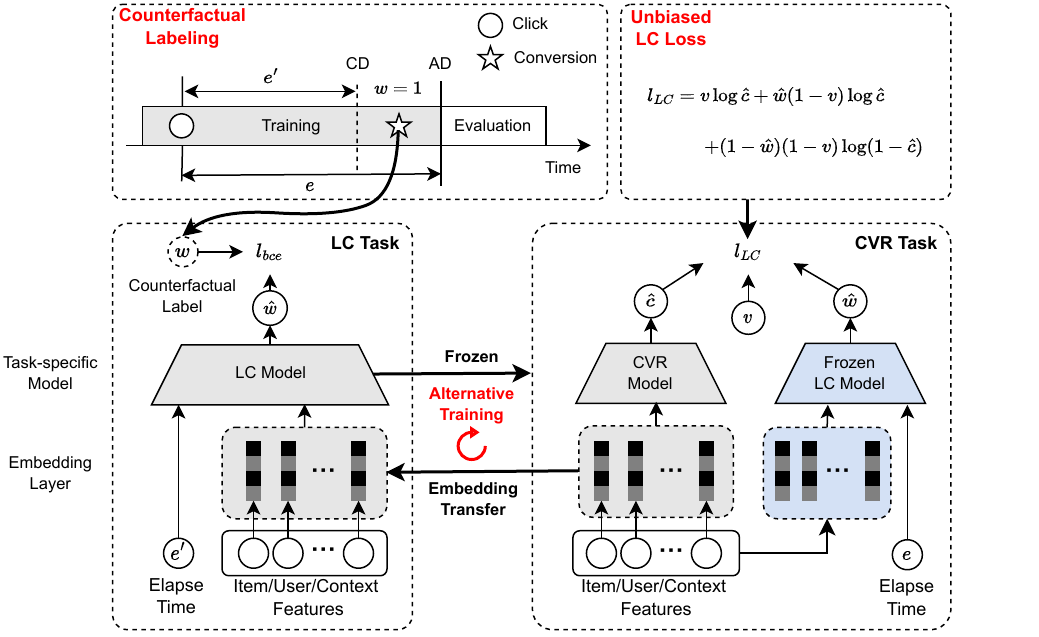}
  \caption{Illustration for our proposed framework ULC.}
  \label{model_fig}
\end{figure*}

We propose an Unbiased delayed feedback Label Correction framework (ULC), which aims to address the delay feedback problem in CVR prediction through label correction. The key idea is that delayed feedback leads to and only leads to incorrect labels. If we are able to identify all the incorrect labels and correct them, we can directly calculate the oracle loss.

Fig.\ref{model_fig} illustrates the overall framework of ULC, which consists of a label correction (LC) model and a CVR prediction model. The LC model is designed to predict the probability that an observed unconverted training sample will finally convert, which is used to calculate our proposed label-corrected loss for CVR model training. We prove in Section 4.2 that if the LC model is accurate enough, the label-corrected loss is an unbiased estimate of the oracle loss. 

The next question is how to learn an accurate LC model. As there is no ready training data for the LC model, we leverage counterfactual labeling to generate training data. It constructs the artificial data by imagining a counterfactual data collection time $T' < T$, the details of which will be introduced in Section 4.3. However, counterfactual labeling suffers from some problems, such as inadequate utilization of the whole training data, which we analyze in Section 4.4. To mitigate this problem, we further apply alternative training to re-train these two models, enhancing the performance by transferring the underlying representation of the CVR model to the LC model.

Next, we elaborate on our framework from unbiased loss, data generation with counterfactual labeling, and alternative training, respectively.

\subsection{Unbiased Loss via Label Correction}
If we have a label correction model $g(\cdot;\psi)$, which can predict the probability $w_i$ of a observed non-conversion sample $i$ to be a final conversion sample, i.e., $P(c_i = 1|x_i,e_i, v_i = 0)$. Then we can directly correct the sample label for the observed negative samples and get the following label-corrected (LC) loss:
\begin{equation}
\begin{aligned}
\mathcal{L}_{LC}(\mathcal{D}) = \frac{1}{|\mathcal{D}|}\sum_{i=1}^{|\mathcal{D}|} [ & v_i \log f(x_i;\theta) + w_i (1 - v_i) \log f(x_i; \theta) \\ &+ (1 - w_i)(1 - v_i) \log (1 - f(x_i; \theta)) ]
\end{aligned}
\end{equation}

The above loss is unbiased to the oracle loss if we have an ideal label correction model, as shown in the following theorem.

\begin{theorem}
    If an ideal label correction model is satisfied, i.e., $w_i = P(c_i = 1|x_i, e_i, v_i = 0)$, then the LC loss is unbiased to the oracle loss.
\end{theorem}

\begin{proof} For simplicity, we denote $\log f(x;\theta)$ as $l_\theta$ and $\log (1 - f(x;\theta))$ as $l_{-\theta}$.
\begin{equation*}
\begin{aligned}
    E[\mathcal{L}_{LC}] &= \iint p(x,e) [p(v = 1|x,e) l_\theta + p(c = 1|v=0,x,e)p(v=0|x,e) l_\theta \\ &+ p(c = 0|v=0,x,e)p(v=0|x,e) l_{-\theta})] dedx \\ 
    &=\iint p(x,e) [p(v = 1,c = 1|x,e) l_\theta + p(c = 1,v=0|x,e) l_\theta \\ 
    &+ p(c = 0,v=0|x,e) l_{-\theta})] dedx \\
   &=\iint p(x,e) [p(c = 1|x,e) l_\theta + p(c = 0|x,e) l_{-\theta})] dedx \\
 &=\int p(x) [p(c=1|x) l_\theta + p(c = 0|x) l_{-\theta}))] dx \\ &= E[\mathcal{L}_{oracle}]
\end{aligned}
\end{equation*}
\end{proof}

The advantage of LC loss over the previous unbiased loss (e.g., FSIW and nnDF) is that it directly complements the information of the correct sample corresponding to the false negative samples, i.e., $w_i(1 - v_i)\log f(x_i;\theta)$. The existing unbiased losses complement the corresponding information in indirect ways, which is strongly influenced by data sparsity and data dynamics. For example, in practice, FSIW complements the correct information by increasing the weights of observed positive samples similar to the false negative sample. However, for some fresh false negative samples, there may not exist similar observed positive samples due to the sparsity and dynamics of CVR data. In this case, these indirect methods cannot effectively supplement the corresponding positive sample information, which is important for the delayed feedback problem as false negative samples are often fresh. In contrast, the LC loss can adequately solve this problem as it directly corrects the label and complements the corresponding correct information.

The remaining problem is how to train an accurate LC model, which we introduce in the next section.

\subsection{Data Generation with Counterfactual Labeling}
For the LC model, there is no ready training data. Note that we need samples with $v_i = 0 \ \& \ c_i = 1$ as positive samples and with $v_i = 0 \ \& \ c_i = 0$ as negative samples. However, there are only samples with $v_i = 1 \ \& \ c_i = 1$ and samples with $v_i = 0$ in the original data. To train the LC model, we need to construct artificial samples.

We leverage a counterfactual method \cite{202006} to generate training data for the LC model. First, we imagine that the training data was collected at a \textit{counterfactual deadline} (CD) before the training data's \textit{actual deadline} (AD), i.e., $T$. The time interval $ \tau$ between the CD and the AD is a hyperparameter. Second, the samples that are clicked but have not converted before the CD are collected as training data, together with $e'$ as the elapsed time of these samples at the CD. Third, we treat the samples with conversion between CD and AD as positive samples, i.e., $w=1$, and others as negative samples, i.e., $w = 0$. Obviously, there exist some samples converting after AD are ignored. Nevertheless, as $\tau$ increases, the proportion of these samples keeps getting smaller. The subsequent experiments demonstrate that even a relatively short $\tau$ can effectively alleviate the delayed feedback problem. 

After data generation, we can train the LC model via the classical binary cross-entropy loss. Then the LC model is frozen and utilized to infer $w$ in the above LC loss. Note that the elapsed time $e$ at the AD is used instead of $e'$ when inferring for LC loss. The detailed data generation procedure is shown in Algorithm \ref{alg2} (lines 2-14).

\subsection{Alternative Training}
Although the training data required for the LC model can be constructed by counterfactual labeling, this method still has some drawbacks. First, data generation only leverages partial training data, i.e., samples that are clicked before CD and converted after CD, which may result in the suboptimal performance of LC model. Second, the LC model also suffers somewhat from delayed feedback. Some potential positive samples that have a long delay and convert after AD may be mistreated as negative samples. Next, we propose an alternative learning based approach to alleviate the first problem. The solution to the second problem we leave to future work.

Note that the CVR prediction model is trained on the whole data, and the conversion rate $P(c = 1|x)$ is similar to the label correction rate $P(c = 1|x,e, v = 0)$. We suppose that the bottom representation (i.e., the embedding layer in Fig.\ref{model_fig}) learned by the CVR model may facilitate the learning of the LC model and alleviate the first problem mentioned above. Therefore, we adopt an alternative learning paradigm. After training the CVR prediction model, the bottom representation of the LC model is initialized using the bottom representation of the CVR prediction model, and then the LC model is retrained. The retrained LC model can be further used for the retraining of the CVR model. The alternative training procedure is shown in Algorithm \ref{alg2}.

\begin{algorithm}
	\renewcommand{\algorithmicrequire}{\textbf{Input:}}
	\renewcommand{\algorithmicensure}{\textbf{Output:}}
	\caption{Alternative training with data generation}
	\label{alg2}
	\begin{algorithmic}[1]
            \REQUIRE training data $\mathcal{D} = \{(x_i, v_i, e_i, cts_i, cvt_i)\}$, $T$ is the timestamp when the data $\mathcal{D}$ are collected, where $x_i$ is the feature vector, $v_i$ is the observed conversion label, $e_i$ is elapsed time since the click timestamp $cts_i$, $cvt_i$ is the conversion timestamp. $\tau$ is a hyperparameter denoting the time interval between CD and AD. $n$ is the rounds of alternative training.
            \ENSURE CVR prediction model $M_{cvr}$
		\STATE Initialize CVR prediction model $M_{cvr}$ and LC model $M_{lc}$
            \STATE // \textit{start to generate data $\mathcal{D}_{lc}$ for LC model training}
		\STATE $\mathcal{D}_{lc} = \emptyset$
            \FOR{$i=1$ to $|\mathcal{D}|$}
                \IF{$cts_i < T - \tau$ and ($v_i == 0$ or $cvt_i > T - \tau$)}
                    \IF{$cvt_i > T - \tau$}
                        \STATE Label the sample $i$ as $w_i = 1$
                    \ELSE
                        \STATE Label the sample $i$ as $w_i = 0$
                    \ENDIF
                    \STATE Insert the sample $(x_i, e_i - \tau, w_i)$ to $\mathcal{D}_{lc}$
                \ENDIF
            \ENDFOR
            \STATE // \textit{finish data generation, start to alternative training}
            \FOR{$i = 1$ to $n + 1$}
                \STATE Initialize $M_{cvr}$
                \STATE Initialize $M_{lc}$ except bottom embeddings
                \STATE Training $M_{lc}$ on $\mathcal{D}_{lc}$ until converge
                \STATE Training $M_{cvr}$ on $\mathcal{D}$ using LC loss until converge
                \STATE Transfer the bottom embeddings from $M_{cvr}$ to $M_{lc}$
            \ENDFOR
            \RETURN CVR prediction model $M_{cvr}$
	\end{algorithmic}  
\end{algorithm}

There are some alternatives compared to alternative training with embedding transfer. For example, joint learning is also a common learning paradigm that enables the LC model to leverage the knowledge of the CVR model. Moreover, in addition to utilizing the learned representation of the CVR model, another easily thought of option is to leverage its prediction. It is possible to mine the mislabeled samples in the training data for the LC model using the CVR prediction model as these potential positive samples might have a high predicted CVR. We also conduct experiments and compare these alternatives in experiment section 5.4.
\section{Experiments}

To validate the effectiveness of our proposed method, we conduct a series of experiments to answer the following research questions:
\begin{itemize}[leftmargin=*]
\item[-] \textbf{RQ1:} How does ULC perform on the CVR prediction task compared to the state-of-the-art methods?
\item[-] \textbf{RQ2:} How do the label correctness and data freshness of counterfactual labeling affect the performance of ULC?
\item[-] \textbf{RQ3:} In addition to embedding-based alternative training, how do other common schemes such as joint learning perform?
\item[-] \textbf{RQ4:} How does the ULC model perform on samples with different delay time?
\end{itemize}

\subsection{Dataset and Settings}
\subsubsection{Datasets} 
To our knowledge, there exists only one public dataset \cite{201401} widely used in the research of the delayed feedback problem in the offline setting. Other public CVR datasets do not have noticeably delayed feedback or lack enough temporal information. Following the common settings in previous work \cite{201401,202006} of using one public and one private dataset, we also introduce a collected private production dataset. 

\textbf{Criteo dataset.} This public dataset contains clicks and the corresponding conversions from Criteo live traffic data. Each sample corresponds to a single click and is described by several categorical features and continuous features, with the corresponding conversion information, if any. It also includes the timestamps of the click and the possible conversion behavior. We use this dataset's last 23 days of data to conduct our experiments. Following previous work \cite{201401, 202006}, three consecutive weeks of data are leveraged as training data, data of the 22nd day is used for validation and the last day is for the testing.
Note that this dataset tracks conversion behavior for each click sample, so the ground truth $c_i$ is available for testing. For validation, we assume that $c_i$ is unknown and use the label-corrected loss for parameter selection. The processed dataset includes 6,363,085 click samples with a conversion rate of 0.2294.

\textbf{Production dataset.} This dataset is collected from a real production platform with game advertising.
In-game payments are treated as conversions. Specifically, we collected and sampled one-month consecutive user feedback logs. The data format is similar to Criteo, and the last two days are used for validation and testing, respectively. The dataset includes over 2,400,000 click samples with a conversion rate of about 0.005.
Statistics are shown in Table \ref{datasets}.

\begin{table}[h]\small
\centering
\caption{Statistics of the datasets.}
\label{datasets}
\begin{tabular}
{p{0.6cm}p{0.4cm}p{0.9cm}p{0.6cm}p{1.3cm}p{0.4cm}p{0.9cm}p{0.6cm}} 
\toprule
\textbf{Name} & \textbf{Days} & \textbf{\#clicks} & 
\textbf{CVR} & | \textbf{Name} & \textbf{Days} & \textbf{\#clicks} & \textbf{CVR}\\ 
\midrule
\textbf{Criteo} & \;\;23& 6,363,085 & 0.2294& |\,\textbf{Production} &\;\;30 & 2,400,000 & 0.0050\\ 
\bottomrule
\end{tabular}
\end{table}

\subsubsection{Evaluation Metrics}
We adopt three metrics that are widely used in CVR prediction tasks \cite{202106, 202103}. The first metric is area under ROC curve (\textbf{AUC}) that measures the pairwise ranking performance of the CVR prediction model. The second is area under the precision-recall curve (\textbf{PRAUC})\cite{202103}, which also measures the pairwise ranking performance.
The third one is the log loss (\textbf{LL}), which measures the accuracy of the absolute value of the CVR prediction. 

To further analyze the benefits gained by solving the delayed feedback problem, we calculate the relative improvements (\textbf{RI}) to the maximum gain (i.e., the improvement of the oracle model over the vanilla model) on the above three metrics. For method $f$, the relative improvements on metric $M(\cdot)$ is defined as $\frac{M(f) - M(Vanilla)}{M(Oralce) - M(Vanilla)}$. Then we can obtain \textbf{RI-AUC}, \textbf{RI-PRAUC} and \textbf{RI-LL}.

\begin{table*}[t]
\centering
\caption{Performance comparisons of the proposed model with baseline models on the public Criteo dataset. The best results are in boldface, and the best baselines are underlined. The superscripts ** indicate $p \le 0.01$ for the t-test of ULC vs. the best baseline. $\uparrow$ means the higher the better, and  $\downarrow$ is for the lower the better.}
\label{main}
\resizebox{0.72\textwidth}{!}{%
\begin{tabular}{l|l|cccccc}
\toprule
\textbf{Backbone} & \textbf{Method}     & \textbf{AUC $\uparrow$}                  & \textbf{PRAUC $\uparrow$}                & \textbf{LL $\downarrow$}                    & \textbf{RI-AUC $\uparrow$}               & \textbf{RI-PRAUC $\uparrow$}             & \textbf{RI-LL $\uparrow$}  \\ \midrule
\multirow{6}{*}{\textbf{MLP}} & \textbf{Vanilla}   & 0.8208                         & 0.6356                         & 0.4505                         & 0.000                          & 0.000                          & 0.000          \\ 
& \textbf{Oracle}   & 0.8441                         & 0.6595                         & 0.4025                         & 1.000                 & 1.000                          & 1.000          \\ \cline{2-8}
& \textbf{DFM}       & 0.8264                         & 0.6398                        & 0.4378                        & 0.2403                         & 0.1757                         & 0.2645         \\
& \textbf{FSIW}      & \underline{0.8335}                   & \underline{0.6465}                   & \underline{0.4178}                   & \underline{0.5450}                   & \underline{0.4560}                   & \underline{0.6812}   \\
& \textbf{nnDF} & 0.6859 & 0.4169 & 0.5969 & -5.789 & -9.150 & -3.05         \\  \cline{2-8}
& \textbf{ULC(ours)}     & \textbf{0.8403**}              & \textbf{0.6543**}              & \textbf{0.4104**}              & \textbf{0.8369**}                       & \textbf{0.7824**}              & \textbf{0.8354**}       \\ \hline
\multirow{6}{*}{\textbf{DeepFM}} & \textbf{Vanilla}   & 0.822                        & 0.6379                         & 0.4451                        & 0.000                          & 0.000                          & 0.000          \\ 
& \textbf{Oracle}   & 0.8442                         & 0.66                         & 0.4023                         & 1.000                 & 1.000                          & 1.000          \\ \cline{2-8}
& \textbf{DFM}       & 0.8266                         & 0.6416                         & 0.4319                        & 0.2072                         & 0.1674                         & 0.3084         \\
& \textbf{FSIW}      & \underline{0.8326}                   & \underline{0.6451}                   & \underline{0.4195}                   & \underline{0.4774}                   & \underline{0.3257}                   & \underline{0.5981}   \\
& \textbf{nnDF} & 0.6994 & 0.4332 & 0.596 & -5.522 & -9.262 & -3.525         \\ \cline{2-8}
& \textbf{ULC(ours)}     & \textbf{0.8393**}              & \textbf{0.6525**}              & \textbf{0.4104**}              & \textbf{0.7792**}                       & \textbf{0.6606**}              & \textbf{0.8107**}       \\ \hline
\multirow{6}{*}{\textbf{AutoInt}} & \textbf{Vanilla}   & 0.8232                         & 0.6383                         & 0.4468                         & 0.000                          & 0.000                          & 0.000          \\ 
& \textbf{Oracle}   & 0.8442                        & 0.6601                         & 0.4025                        & 1.000                 & 1.000                          & 1.000          \\ \cline{2-8}
& \textbf{DFM}       & 0.8276                         & 0.6412                         & 0.4306                         & 0.2095                         & 0.1330                         & 0.3656         \\
& \textbf{FSIW}      & \underline{0.8329}                   & \underline{0.646}                   & \underline{0.4192}                   & \underline{0.4619}                   & \underline{0.3532}                   & \underline{0.6230}   \\
& \textbf{nnDF} & 0.6903 & 0.4214 & 0.6123 & -6.328 & -9.949 & -3.735         \\ \cline{2-8}
& \textbf{ULC(ours)}     & \textbf{0.8388**}              & \textbf{0.6517**}              & \textbf{0.4114**}              & \textbf{0.7428**}                       & \textbf{0.6146**}              & \textbf{0.7990**}       \\ \hline
\multirow{6}{*}{\textbf{DCNV2}} & \textbf{Vanilla}   & 0.8229                         & 0.6386                         & 0.4454                         & 0.000                          & 0.000                          & 0.000          \\ 
& \textbf{Oracle}   & 0.8447                         & 0.6606                         & 0.4018                         & 1.000                 & 1.000                          & 1.000          \\ \cline{2-8}
& \textbf{DFM}       & 0.8272                         & 0.6417                        & 0.4325                         & 0.1972                         & 0.1409                         & 0.2958         \\
& \textbf{FSIW}      & \underline{0.8328}                   & \underline{0.6461}                   & \underline{0.4174}                   & \underline{0.4541}                   & \underline{0.3409}                   & \underline{0.6422}   \\
& \textbf{nnDF} & 0.6867 & 0.423 & 0.6101 & -6.247 & -9.8 & -3.777         \\ \cline{2-8}
& \textbf{ULC(ours)}     & \textbf{0.8391**}              & \textbf{0.6519**}              & \textbf{0.4105**}              & \textbf{0.7431**}                       & \textbf{0.6045**}              & \textbf{0.8004**}       \\
\bottomrule
\end{tabular}%
}
\end{table*}

\subsubsection{Compared Methods}
The following state-of-the-art methods are our baselines for solving delayed feedback in CVR prediction:

\begin{itemize}[leftmargin=*]
\item \textbf{Vanilla}: a CVR model trained with the observed conversion label. This is the lower bound of possible improvements.
\item \textbf{Oracle}: a CVR model trained with the ground truth label instead of observed labels. This is the upper bound of possible improvements.
\item \textbf{DFM}\cite{201401}: a CVR model trained using delayed feedback loss.
\item \textbf{FSIW}\cite{202006}: a CVR model trained using FSIW loss and pre-trained auxiliary models.
\item \textbf{nnDF}\cite{202204}: a CVR model trained using the nnDF loss.
\item \textbf{ULC(ours)}: a CVR model alternately trained with the LC model and the LC loss.
\end{itemize}

The above methods can be applied to different CVR models. Due to the column anonymity of the Criteo dataset, we cannot use the models that rely on user modeling. We consider the following classical models that focus on feature interactions as backbones:

\begin{itemize}[leftmargin=*]
\item \textbf{MLP}: the classical fully connected neural networks.
\item \textbf{DeepFM} \cite{deepfm}: a model combining the factorization machines and deep neural networks.
\item \textbf{AutoInt} \cite{autoint}: a model using multi-head self-attention to learn the high-order feature interactions automatically.
\item \textbf{DCNV2} \cite{dcnv2}: a model using deep and cross networks to learn effective explicit and implicit feature crosses.
\end{itemize}

\subsubsection{Implementation Details}
The embedding size is 64 for all the methods. The MLP model in all the backbones is a simple three-layer model with hidden units [256,256,128] and Leaky ReLU activation. For AutoInt, the layer number is 3, the number of heads is 2, and the attention size is 64. For DCNV2, we use the stacked structure and one cross-layer. Adam \cite{adam} is used as the optimizer, and the learning rate is tuned in the range of [1e-3, 5e-4, 1e-4] with L2 regularization tuned in [0, 1e-7, 1e-6, 1e-5, 1e-4]. The batch size is set to 1024 for all the methods except nnDF. Given that the nnDF approach cannot apply to batch-wise training, we set the batch size to the size of the whole training set. For a fair comparison, we consistently use the MLP as the auxiliary model for all methods that rely on auxiliary models. The additional hyperparameters for the baselines are fine-tuned. Early stopping is applied to obtain the best parameters. We repeat each experiment 5 times with different random seeds and report the average results and make the statistical tests. \footnote{The codes can be found at https://github.com/yfwang2021/ULC. A mindspore version: https://gitee.com/mindspore/models/tree/master/research/recommend/ULC.}

\subsection{Overall Performance: RQ1}
From Table \ref{main}, we can observe that our proposed method ULC outperforms all the baselines and achieves state-of-the-art performance on all the backbones. There are some further observations. First, the oracle model works significantly better than the vanilla model, which validates that the delayed feedback problem indeed hurts the performance of CVR model. Second, FSIW performs significantly better than the DFM and Vanilla models, which is consistent with previous studies \cite{202006}. However, nnDF is significantly weaker than Vanilla method. It is because nnDF loss requires global dependence computation and can only be optimized using full training data when updating, which leads to a weaker performance than batch-wise optimization methods. Third, compared to the best baseline, our method shows a significant improvement of 0.76\% in the AUC metric, 1.02\% in the PRAUC metric, and 1.85\% in the LL metric on average across the four backbones, which demonstrates the effectiveness of our proposed method.

We further analyze the benefits gained by solving the delayed feedback problem. 
As shown in Table \ref{main}, our method narrows the gap between Vanilla and Oracle by 77.55\% in the AUC metric, 66.55\% in the PRAUC metric, and 83.13\% in the LL metric on average across the four backbones. Compared to the best baseline, our method shows a significant improvement of 60.3\% in the RI-AUC metric, 81.43\% in the RI-PRAUC metric, and 27.76\% in the RI-LL metric on average across the four backbones. This shows that our method can effectively alleviate the delayed feedback problem.

Fig.\ref{private_performance} shows the offline performance on the production dataset. For limited space, we only present the results with MLP as the backbone. It is clearly observed that our proposed method alleviates the delayed feedback problem and outperforms the two best baselines.

To guarantee the reproducibility of our work, and also due to the page limitation, we make the following further detailed analyses on the public-available dataset.

\begin{figure}[h]
  \centering
  \includegraphics[width=\linewidth]{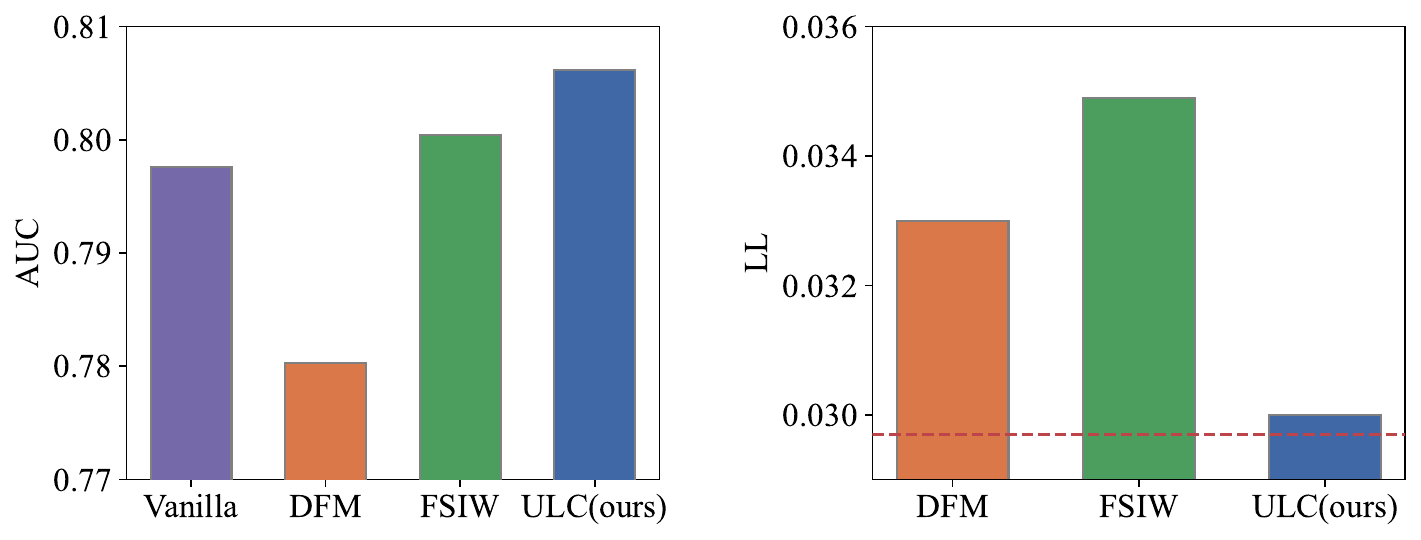}
 \caption{Performance comparisons of the proposed method with the top two baselines on the private dataset. The backbone model is MLP. The red dotted line in the right figure denotes Oracle.
 }
  \label{private_performance}
\end{figure}

\subsection{Analysis on Counterfactual Labeling: RQ2}
\begin{figure*}[t]
  \centering
  \includegraphics[width=15cm]{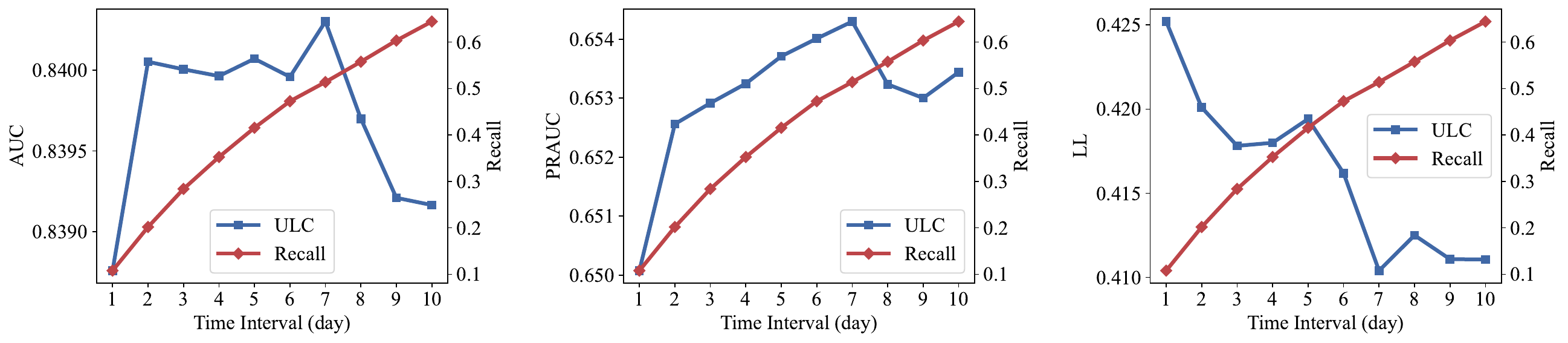}
  \caption{Effect of different time intervals between counterfactual deadline (CD) and actual deadline (AD) on the Criteo dataset with MLP as the backbone. The blue line represents the performance of ULC and the red line represents the recall of counterfactual labeling on positive samples. Larger recall means fewer mislabels in the training data of LC model.}
  \label{time_interval}
\end{figure*}

In counterfactual labeling, only the samples converted between CD and AD are treated as positive samples (i.e., $w = 1$), which leads to some samples converted after AD being mislabeled as negative samples. A long time interval between CD and AD can improve label correctness of counterfactual labeling but reduce data freshness as only clicked data before CD are utilized. We further analyze the effect of different time intervals.

Experimental results using different time intervals are shown in Fig.\ref{time_interval}. First, increasing the time interval can effectively increase the recall of counterfactual labeling on positive samples. Second, the best $\tau$ on the Criteo dataset is around a week. Besides, smaller or larger $\tau$ will reduce the performance of CVR model. Smaller $\tau$ leads to more mislabeled samples in the training data of LC model, which in turn leads to lower performance of CVR model. Larger $\tau$ values, while reducing the mislabeled samples, will make the training data of the LC model older, which leads to its inability to correct well for false negative samples in the CVR training data, since these false negative samples are relatively fresh.

\subsection{Effectiveness of Alternative Training: RQ3}
In addition to embedding-based alternative training, there are also some other design schemes that enable the LC model to exploit the knowledge of CVR prediction model. We conduct experiments on these schemes.

\subsubsection{Joint Learning Strategy}
An obvious solution is to jointly train the LC model and the CVR model, which also enables the LC model to utilize the information learned by the CVR model. We use a simple shared-bottom structure \cite{mmoe} to validate the effectiveness of this scheme. The joint loss is a linear weighting of the LC model loss and the CVR model loss.

\subsubsection{Prediction-based Alternative Training Strategy}
In alternative training, in addition to using the learned representation of the CVR model, another easily thought of option is to leverage its prediction. Note that in Section 4.3, we mention that some potential positive samples that have a long delay and convert after AD may be mislabeled as negative samples during counterfactual labeling. To alleviate this problem, we consider using the prediction of the CVR model to mine these potentially positive samples. Intuitively, samples with high predicted CVR are more likely to be potentially positive samples. Thus, we design three simple strategies to process the training data of the LC model: (i) hard strategy, i.e., negative samples ($w = 0$) with predicted CVR above a predefined threshold are treated as positive samples ($w = 1$); (ii) soft strategy, i.e., using predicted CVR as the label for each negative sample; (iii) drop strategy, i.e., dropping negative samples with predicted CVR above a predefined threshold from the training data of the LC model, as the labels of these samples are not reliable.

\subsubsection{Comparisons on Different Strategies}
The results of the above strategies on Criteo with MLP as backbone are shown in Fig.\ref{alternatives}. Results on AUC are similar to PRAUC and hence omitted. We have the following observations:
(i) using a simple joint learning scheme cannot improve performance. Instead, there is a large loss of performance. The reason is that the inaccurate LC model at the early training stage will mislead the CVR model, which in turn affects the subsequent training. (ii) the three prediction-based strategies cannot improve the performance of the CVR model and even cause a slight degradation. The potential positive samples after AD have a higher delay than $\tau$, and the number of these samples is very small compared to the number of true negative samples (about 1:50 ratio). Using only predicted CVR cannot effectively discover these samples; instead, it introduces noise.

\begin{figure}[h]
  \centering
  \includegraphics[width=8cm]{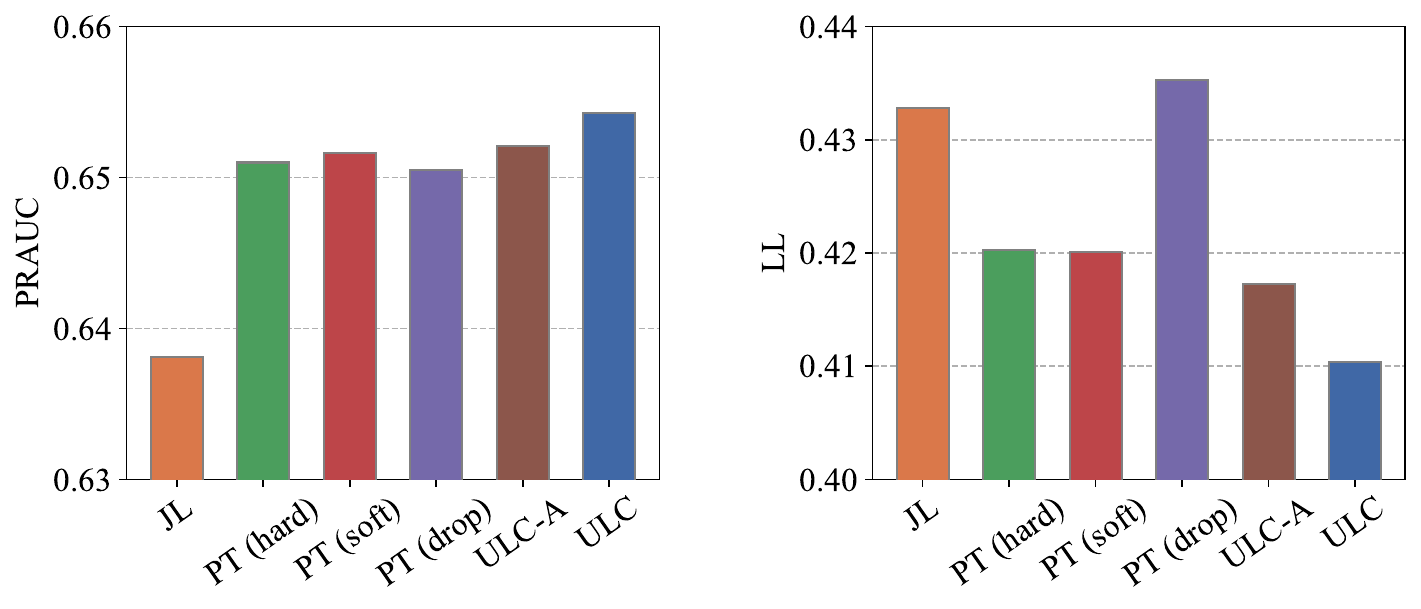}
 \caption{Comparisons on different training strategies. Dataset: Criteo. Backbone: MLP. "ULC-A": Proposed ULC without alternative training. "JL": Joint learning. "PT": Prediction-based alternative training. 
 }
  \label{alternatives}
\end{figure}

\begin{figure}[t]
  \centering
  \includegraphics[width=\linewidth]{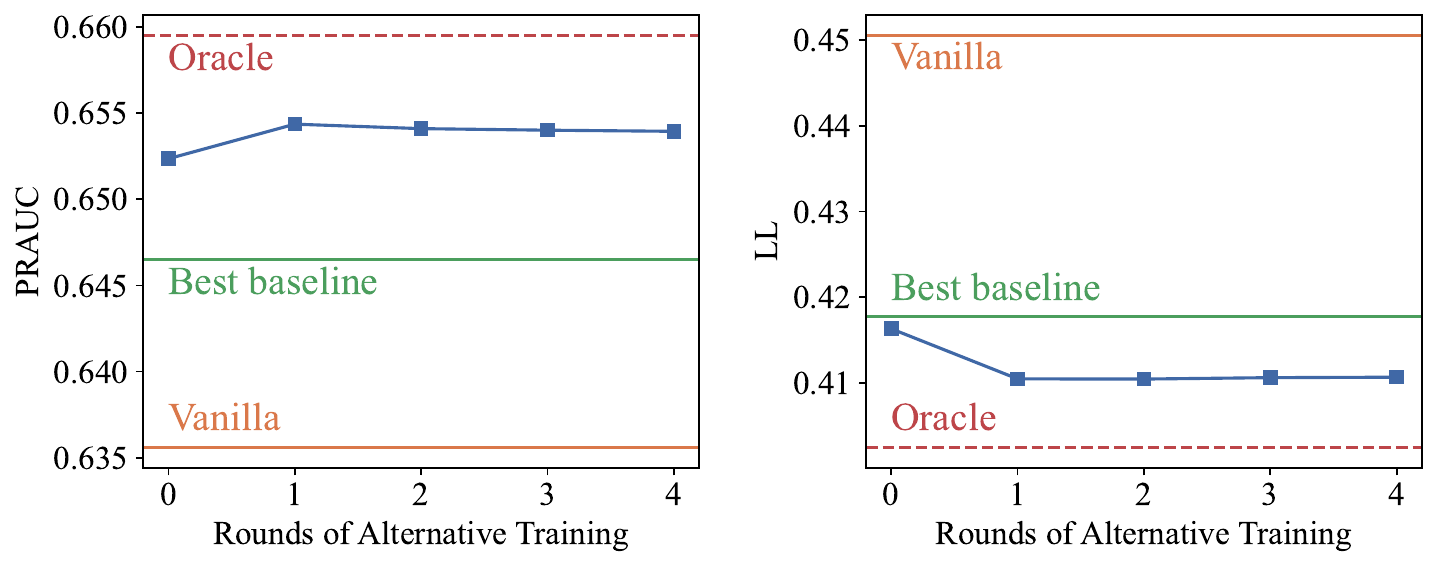}
  \caption{Performance w.r.t. the alternative training rounds. Dataset: Criteo. Backbone: MLP. Note that $0$ on the $x$ axis means no alternative training.}
  \label{alt_number}
\end{figure}

\subsubsection{Sensitivity of Alternative Training Rounds}
We further analyze the effect of alternative training rounds on the performance of ULC. $n$ controls the rounds of alternative training, and $n=0$ means no alternative training. We conduct experiments using different values of $n$ on the Criteo dataset with MLP as the backbone. As shown in Fig.\ref{alt_number}, alternative training once can significantly improve the performance of the CVR prediction model, which validates the effectiveness of alternative training. Besides, one round is enough, and more rounds have little impact, which is reasonable since the first round that changes the initialization of the LC model from random to the embeddings of CVR model brings more significant changes than the subsequent rounds. Note that even without alternative learning, the performance of ULC is still significantly better than the best baseline, which reflects the effectiveness of using the LC model for label correction. Results on AUC are similar to PRAUC and hence omitted.

\subsection{Analysis on Different Delay Time: RQ4}
The delay time is an important property of delayed feedback. We further analyze the performance of CVR model and LC model on samples with different delay time.

\subsubsection{CVR performance on different delay time}
For a fair comparison, we divide the positive samples in the test set into five groups in ascending order based on their delay time. Each group has the same number of positive samples. Then, each group is combined with all the negative samples in the test set to form test sets with different delay times. In this way, the number of positive and negative samples is the same for different test sets. Further, since the log loss is sensitive to the conversion rate, to ensure that the conversion rate in the test set is consistent with the original test set, we duplicate five copies of each positive sample.

Experiment results on the Criteo dataset with MLP as backbone are shown in Fig.\ref{delay_interval}. We have the following observations: 
(i) for the Oracle model without the delayed feedback problem, its performance decreases somewhat as the sample delay time increases, which indicates that samples with a long delay time are more likely to be hard samples. (ii) as the sample delay time increases, the Oracle model performs increasingly better than Vanilla, which is because positive samples are more likely to be false negative samples as the delay time increases. (iii) our method significantly outperforms the Vanilla model and the best baseline on samples with high delays (e.g., G3, G4, and G5), and our boost increases as the delay time increases, which reflects the effectiveness of our method.

\begin{figure}[t]
  \centering
  \includegraphics[width=\linewidth]{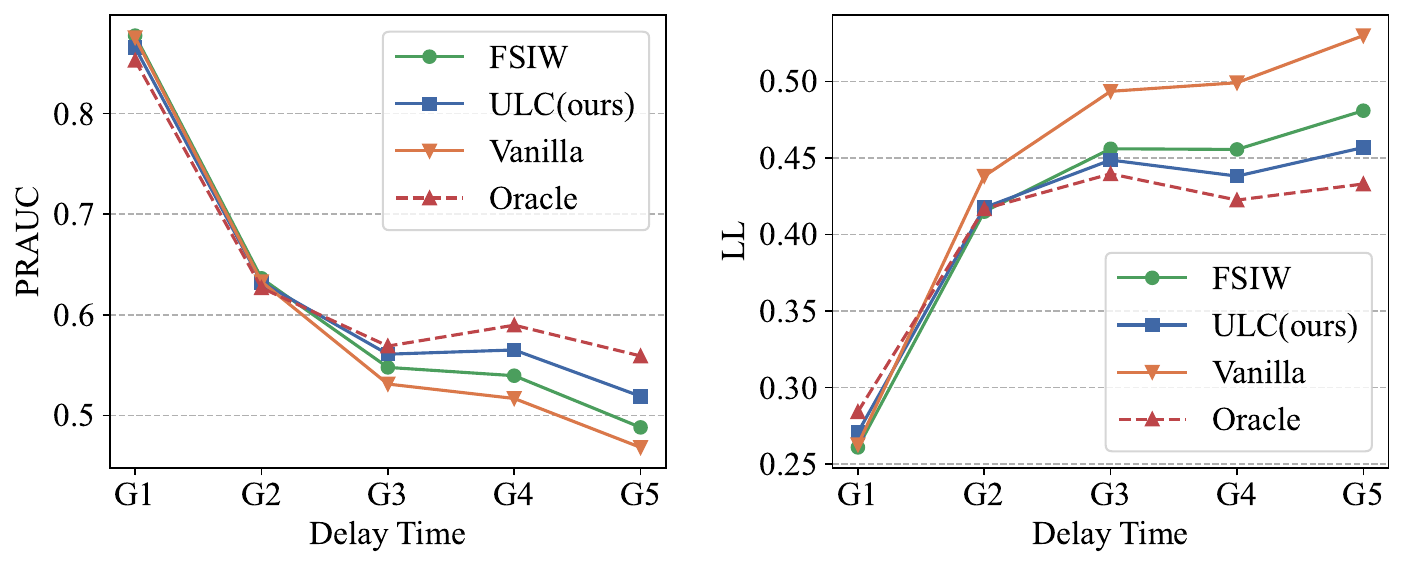}
  \caption{Performance w.r.t. test samples with different delay time on the Criteo dataset with MLP as the backbone. A larger value of the $x$ axis means a longer delay time. Results on AUC are similar to PRAUC and hence omitted.}
  \label{delay_interval}
\end{figure}

An interesting phenomenon is that the Vanilla model performs better than the Oracle model on samples with short delays (G1). It may be because samples with short delays have a higher percentage of observed positive samples than actual positive samples. Further analysis can be found in Appendix.

\subsubsection{LC performance on different delay time}
We further analyze the performance of the LC model on samples with different delay time. Similarly, we divide the false negative samples in the training data into five groups in ascending order based on their delay time. Each group has the same number of false positive samples. Then, as the goal of the LC model is to distinguish between false negative samples and true negative samples, each group is combined with all the true negative samples in the training data to form evaluation data with different delay time.

\begin{table}[h]
\centering
\caption{Label correction performance of ULC w.r.t. samples with different delay time. G5 is the group with the longest delay, and G1 has the shortest delay. $\uparrow$ means the higher the better, and  $\downarrow$ is for the lower the better.}
\label{lc_performance}
\resizebox{0.4\textwidth}{!}{%
\begin{tabular}{l|ccccc}
\toprule
Metrics & G1    & G2    & G3     & G4    & G5     \\ \midrule
AUC $\uparrow$    & 0.8698 & 0.8350 & 0.8117 & 0.7896 & 0.7811 \\
PRAUC $\uparrow$  & 0.1397 & 0.0757 & 0.0545 & 0.0434 & 0.0398 \\
LL  $\downarrow$    & 0.0549 & 0.0584 & 0.0604 & 0.0621 & 0.0628 \\ \bottomrule
\end{tabular}%
}
\end{table}

Experiment results on the Criteo dataset are shown in Table \ref{lc_performance}. We have the following observations: (i) AUC ranges from 0.7811 to 0.8698 at different delay time, which reflects that the LC model can effectively recognize false negative samples from all negative samples. (ii) the performance of the LC model decreases as the delay time of the false negative samples increases. There are two possible reasons for this. First, samples with longer delays are more likely to be hard samples. Second, in counterfactual labeling, false negative samples with longer delays are more likely to convert after AD and be recognized as true negative examples, which damages the performance of LC model.
\section{Conclusions and Future Work}
In this paper, we propose a framework ULC to address the delayed feedback problem in the offline setting via unbiased label correction. The key idea is that delayed feedback leads to and only leads to incorrect labels. If the incorrect labels can be effectively corrected, the delayed feedback problem can be well addressed. ULC uses an additional LC model to guide the CVR prediction model for unbiased label correction and enhances the performance through alternative training. We prove theoretically that the label-corrected loss in our method is an unbiased estimate of the oracle loss. Comparative experiments on both public and private datasets and detailed analyses show that ULC effectively alleviates the delayed feedback problem and consistently outperforms the previous state-of-the-art methods.

For future work, we are interested in the following points. First, using multiple and dynamic counterfactual deadlines is likely to exploit training data more effectively. Second, given that samples with a long delay time are more likely to be hard samples, we would like to design approaches to enhance the model performance on long-delay samples. Third, we are interested in the combination of our method to selection bias in CVR prediction.

\begin{acks}
This work is supported by the Natural Science Foundation of China (Grant No.U21B2026), the fellowship of China Postdoctoral Science Foundation (No.2022TQ0178) and Huawei (Huawei Innovation Research Program). We also thank MindSpore for the partial support of this work, which is a new deep learning computing framework.
\end{acks}

\bibliographystyle{ACM-Reference-Format}
\balance
\bibliography{sample-base}

\appendix
\section{Appendix}
\subsection{Further Analysis on Different Delay Time}
An interesting phenomenon in Figure.\ref{delay_interval} is that the Vanilla model performs better than the Oracle model on samples with short delays (G1). It may be because samples with short delays have a higher percentage of observed positive samples than actual positive samples.
Specifically, the proportion of G1 to the observed positive samples is 25.8\%, which is higher than that of G1 to the true positive samples, 21.4\%. Thus, the Vanilla model focuses more on these short-delay samples and learns better about them. We analyze the training samples with different delay time in the same way as above. As shown in Fig.\ref{delay_time_train}, the Vanilla model indeed learns better on the short-delay samples of training data than the Oracle model. Moreover, we can also observe that the samples with a longer delay time on the training set have a larger training loss for the Oracle model. As the number of true positive samples is the same for these groups, this also indicates that samples with longer delays are more likely to be hard samples.

\begin{figure}[h]
  \centering
  \includegraphics[width=\linewidth]{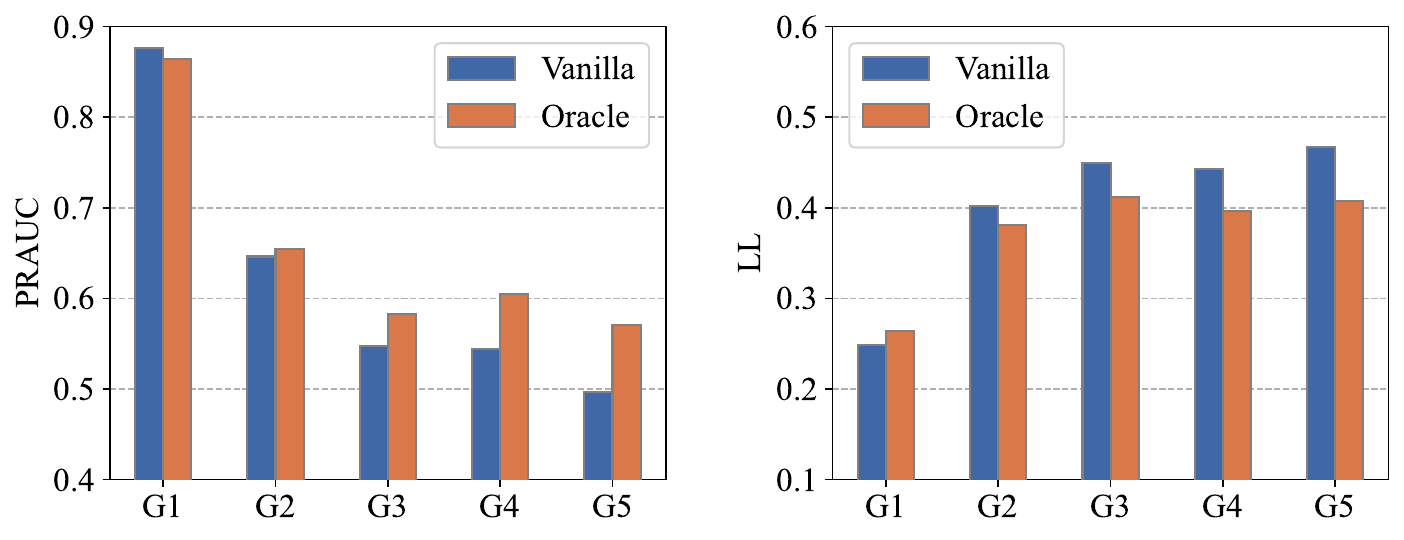}
  \caption{Performance w.r.t. training samples with different delay time on the Criteo dataset with MLP as the backbone. A larger value of the $x$ axis means a longer delay time. Results on AUC are similar and hence omitted.}
  \label{delay_time_train}
\end{figure}

\subsection{Case Study}
Here are two concrete cases on the Criteo dataset with MLP as the backbone. For each test sample, we find the ten most similar samples in the training set, as the label correctness of these training samples might have a strong impact on the prediction of the test sample. We calculate similarity using the L2 distance of sample embeddings in our model. As shown in Table \ref{case1} and \ref{case2}, it can be found that our method effectively corrects the labels of false negative samples without wrongly categorizing true negative samples as positive ones. Consequently, the prediction accuracy of the test sample is enhanced.

\begin{table}[h]
\centering
\caption{The first concrete case on the Criteo dataset with MLP as the backbone. The false negatives are in boldface. Our method effectively corrects the labels of false negative samples (5063626, 5630493). Meanwhile, the true negative samples (1738500, 904961, 3478664) are not corrected to positive samples.}
\label{case1}
\resizebox{0.45\textwidth}{!}{%
\begin{tabular}{cccc}
\toprule \toprule
\textbf{Test Sample ID}                                                               & \textbf{Label}      & \textbf{Vanilla Pred.}  & \textbf{ULC(ours) Pred.} \\ \midrule
\textbf{186833}                                                                       & 1          & 0.081          & 0.787           \\ \midrule \midrule
\begin{tabular}[c]{@{}c@{}}\textbf{Ten Most Similar} \\ \textbf{Training Samples}\end{tabular} & \textbf{True Label} & \textbf{Observed Label} & \textbf{Corrected Label} \\ \midrule
\textbf{5063626}                                                             & \textbf{1} & \textbf{0}     & \textbf{0.906}  \\
\textbf{4136975}                                                             & 1          & 1              & 1               \\
\textbf{3726489}                                                             & 1          & 1              & 1               \\
\textbf{1738500}                                                             & 0          & 0              & 0.000           \\
\textbf{5630493}                                                             & \textbf{1} & \textbf{0}     & \textbf{0.821}  \\
\textbf{904961}                                                              & 0          & 0              & 0.023           \\
\textbf{1980991}                                                             & 1          & 1              & 1               \\
\textbf{3424953}                                                             & 1          & 1              & 1               \\
\textbf{3478664}                                                             & 0          & 0              & 0.044           \\
\textbf{5645863}                                                             & 1          & 1              & 1      \\ \bottomrule \bottomrule     
\end{tabular}%
}
\end{table}

\begin{table}[h]
\centering
\caption{The second concrete case on the Criteo dataset with MLP as the backbone. The false negatives are in boldface. Our method effectively corrects the labels of false negative samples (5782400, 5796674, 5801703). Meanwhile, the true negative samples (remaining part) are not corrected to positive samples.}
\label{case2}
\resizebox{0.45\textwidth}{!}{%
\begin{tabular}{cccc}
\toprule \toprule
\textbf{Test Sample ID}                                                               & \textbf{Label}                    & \textbf{Vanilla Pred.}            & \textbf{ULC(ours) Pred.}              \\ \midrule
\textbf{146531}                                                                                & 1                                 & 0.061                             & 0.790                                 \\ \midrule \midrule
\textbf{\begin{tabular}[c]{@{}c@{}}Ten Most Similar \\ Training Samples\end{tabular}} & \textbf{True Label}               & \textbf{Observed Label}           & \textbf{Corrected Label}              \\ \midrule
\textbf{5782400}                                                                      & {\color[HTML]{2C3A4A} \textbf{1}} & {\color[HTML]{2C3A4A} \textbf{0}} & {\color[HTML]{2C3A4A} \textbf{0.538}} \\
\textbf{5796674}                                                                      & {\color[HTML]{2C3A4A} \textbf{1}} & {\color[HTML]{2C3A4A} \textbf{0}} & {\color[HTML]{2C3A4A} \textbf{0.677}} \\
\textbf{5801703}                                                                      & {\color[HTML]{2C3A4A} \textbf{1}} & {\color[HTML]{2C3A4A} \textbf{0}} & {\color[HTML]{2C3A4A} \textbf{0.616}} \\
\textbf{2648493}                                                                      & 0                                 & 0                                 & 0.002                                 \\
\textbf{4545673}                                                                      & 0                        & 0                        & 0.012                        \\
\textbf{686762}                                                                       & 0                                 & 0                                 & 0.000                                 \\
\textbf{214715}                                                                       & 0                                 & 0                                 & 0.000                                 \\
\textbf{5184723}                                                                      & 0                                 & 0                                 & 0.007                                 \\
\textbf{602570}                                                                       & 0                                 & 0                                 & 0.000                                 \\
\textbf{5382003}                                                                      & 0                                 & 0                                 & 0.013     \\ \bottomrule \bottomrule                           
\end{tabular}%
}
\end{table}

\end{document}